\documentclass[runningheads]{llncs}
\usepackage[T1]{fontenc}
\usepackage{hyperref}
\usepackage{uniinput}
\usepackage[textsize=tiny]{todonotes} 
\usepackage{siunitx}
\usepackage[capitalize]{cleveref}
\usepackage{amsmath}
\usepackage{appendix}

\usepackage{amsthm}

\usepackage{amssymb}
\usepackage{thmtools}
\usepackage{thm-restate}
\usepackage{subcaption}

\DeclareMathOperator{\cheap}{cheap}
\DeclareMathOperator{\expan}{expan}

\usepackage{graphicx}
\usepackage{color}

\begin{document}
\title{Deterministic Performance Guarantees for Bidirectional BFS on Real-World Networks}
\titlerunning{Deterministic Performance for Bi-BFS on Real-World Networks}
%
\author{Thomas Bläsius\inst{1}\orcidID{0000-0003-2450-744X} \and
Marcus Wilhelm\inst{1}\orcidID{0000-0002-4507-0622}}
\authorrunning{T. Bläsius, M. Wilhelm}
%
\institute{Karlsruhe Institute of Technology (KIT), Karlsruhe, Germany\\
  \email{first.last@kit.edu}}
\maketitle              
\begin{abstract}
  A common technique for speeding up shortest path queries in graphs is
  to use a bidirectional search, i.e., performing a forward search
  from the start and a backward search from the destination until a
  common vertex on a shortest path is found. In practice, this has a
  massive impact on performance in some real-world networks,
  while it seems to save only a constant factor in other types of
  networks. Although finding shortest paths is a ubiquitous
  problem, only few studies have attempted to explain the
  apparent asymptotic speedups on some networks using average case
  analysis on certain models of real-world networks.

  In this paper we provide a new perspective on this, by analyzing
  deterministic properties that allow theoretical analysis and that
  can be easily checked on any particular instance. We prove that
  these parameters imply sublinear running time for the bidirectional
  breadth-first search in several regimes, some of which are tight.
  Furthermore, we perform experiments on a large set of real-world
  networks and show that our parameters capture the concept of
  practical running time well.

\keywords{scale-free networks \and bidirectional BFS \and bidirectional shortest paths \and distribution-free analysis.}

{\footnotesize Note: this is the extended version of a paper accepted at IWOCA'23}
\end{abstract}

\section{Introduction}
\label{sec:intro}

A common way to speed up the search for a shortest path between two
vertices is to use a bidirectional search strategy instead of a
unidirectional one.  The idea is to explore the graph from both, the
start and the destination vertex, until a common vertex somewhere in
between is discovered.  Even though this does not improve upon the
linear worst-case running time of the unidirectional search, it leads
to significant practical speedups on some classes of networks.
Specifically, Borassi and Natale~\cite{borassi_kadabra_2016} found
that bidirectional search seems to run asymptotically faster than
unidirectional search on scale-free real-world networks.  This does,
however, not transfer to other types of networks like for example
transportation networks, where the speedup seems to be a constant
factor~\cite{DBLP:series/lncs/BastDGMPSWW16}.

There are several results aiming to explain the practical run times
of the bidirectional search, specifically of the balanced
bidirectional breadth-first search (short: bidirectional BFS). These
results have in common that they analyze the bidirectional BFS on
probabilistic network models with different properties. Borassi and
Natale~\cite{borassi_kadabra_2016} show that it takes $O(\sqrt{n})$
time on Erdős-Rényi-graphs~\cite{er-rgi-59} with high probability.
The same holds for slightly non-uniform random graphs as long as the
edge choices are independent and the second moment of the degree
distribution is finite. For more heterogeneous power-law degree
distributions with power-law exponent in $(2, 3)$, the running time
is $O(n^c)$ for $c∈ [1/2, 1)$. Note that this covers a wide range of
networks with varying properties in the sense that it predicts
sublinear running times for homogeneous as well as heterogeneous
degree distributions. However, the proof for these results heavily
relies on the independence of edges, which is not necessarily given
in real-world networks. Bläsius et al.~\cite{bff-espsf-22} consider
the bidirectional BFS on network models that introduce dependence of
edges via an underlying geometry. Specifically, they show sublinear
running time if the underlying geometry is the hyperbolic plane,
yielding networks with a heterogeneous power-law degree distribution.
Moreover, if the underlying geometry is the Euclidean plane, they
show that the speedup is only a constant factor.

Summarizing these theoretical results, one can roughly say that the
bidirectional BFS has sublinear running time unless the network has
dependent edges and a homogeneous degree distribution.  Note that
this fits to the above observation that bidirectional search works
well on many real-world networks, while it only achieves a constant
speedup on transportation networks.  However, these theoretical
results only give actual performance guarantees for networks
following the assumed probability distributions of the analyzed
network models.  Thus, the goal of this paper is to understand the
efficiency of the bidirectional BFS in terms of deterministic
structural properties of the considered network.

\paragraph{Intuition.}
To present our technical contribution, we first give high-level
arguments and then discuss where these simple arguments fail.  As
noted above, the bidirectional BFS is highly efficient unless the
networks are homogeneous and have edge dependencies.  In the field of
network science, it is common knowledge that these are the networks
with high diameter, while other networks typically have the
small-world property.  This difference in diameter coincides with
differences in the expansion of search spaces.  To make this more
specific, let $v$ be a vertex in a graph and let $f_v(d)$ be the
number of vertices of distance at most $d$ from $v$.  In the
following, we consider two settings, namely the setting of
\emph{polynomial expansion} with $f_v(d) \approx d^2$ and that of
\emph{exponential expansion} with $f_v(d) \approx 2^d$ for all
vertices $v \in V$.  Now assume we use a BFS to compute the shortest
path between vertices $s$ and $t$ with distance~$d$.

To compare the unidirectional with the bidirectional BFS, note that
the former explores the $f_s(d)$ vertices at distance $d$ from $s$,
while the latter explores the $f_s(d / 2) + f_t(d / 2)$ vertices at
distance $d / 2$ from $s$ and $t$.  In the polynomial expansion
setting, $f_s(d / 2) + f_t(d / 2)$ evaluates to
$2(d/2)^2 = d^2/2 = f_s(d) / 2$, yielding a constant speedup of $2$.
In the exponential expansion setting, $f_s(d / 2) + f_t(d / 2)$
evaluates to $2 \cdot 2^{d / 2} = 2 \sqrt{f_s(d)}$, resulting in a
polynomial speedup.

With these preliminary considerations, it seems like exponential
expansion is already the deterministic property explaining the
asymptotic performance improvement of the bidirectional BFS on many
real-world networks. However, though this property is strong enough
to yield the desired theoretic result, it is too strong to actually
capture real-world networks. There are two main reasons for that.
First, the expansion in real-world networks is not that clean, i.e.,
the actual increase of vertices varies from step to step. Second, and
more importantly, the considered graphs are finite and with
exponential expansion, one quickly reaches the graph's boundary where
the expansion slows down. Thus, even though search spaces in
real-world networks are typically expanding quickly, it is crucial to
consider the number of steps during which the expansion persists. To
actually capture real-world networks, weaker conditions are needed.

\paragraph{Contribution.} The main contribution of this paper is
to solve this tension between wanting conditions strong enough to
imply sublinear running time and wanting them to be sufficiently weak
to still cover real-world networks.
We solve this by defining multiple parameters describing expansion
properties of vertex pairs. These parameters address the above issues
by covering a varying amount of expansion and stating requirements on
how long the expansion lasts. We refer to \cref{sec:prelim} and
\cref{sec:expand-search-spac-basic-prop} for the exact technical
definitions, but intuitively we define the \emph{expansion overlap}
as the number of steps for which the exploration cost is growing
exponentially in both directions. Based on this, we give different
parameter settings in which the bidirectional search is sublinear. In
particular, we show sublinear running time for logarithmically sized
expansion overlap (\cref{lem_sublinear_by_large_overlap}) and for an
expansion overlap linear in the distance between the queried vertices
(\cref{lem_sublinear_with_cheap_start}, the actual statement is
stronger). For a slightly more general setting we also prove a tight
criterion for sublinear running time in the sense that the parameters
either guarantee sublinear running time or that there exists a family
of graphs that require linear running time (\cref{thm_dichotomy}).
Note that the latter two results also require the relative difference
between the minimum and maximum expansion to be constant. Finally, we
demonstrate that our parameters do indeed capture the behavior
actually observed in practice by running experiments on more than
\SI{3}{k} real-world networks.

\paragraph{Related work.} Our results fit into the more general
theme of defining distribution-free~\cite{rs-dfmsn-21} properties that capture
real-world networks and analyzing algorithms based on these
deterministic properties.

Borassi, Crescenzi, and Trevisan~\cite{DBLP:conf/soda/BorassiCT17}
analyze heuristics for graph properties such as the diameter and
radius as well as centrality measures such as closeness. The analysis
builds upon a deterministic formulation of how edges form based on
independent probabilities and the birthday paradox. The authors
verify their properties on multiple probabilistic network models as
well as real-world networks.

Fox et al.~\cite{DBLP:journals/siamcomp/FoxRSWW20} propose a
parameterized view on the concept of triadic closure in real-world
networks. This is based on the observation that in many networks, two
vertices with a common neighbor are likely to be adjacent. The
authors thus call a graph $c$-closed if every pair of vertices $u,v$
with at least $c$ common neighbors is adjacent. They show that
enumerating all maximal cliques is in FPT for parameter $c$ and also
for a weaker property called weak $c$-closure. The authors also
verify empirically that real-world networks are weakly $c$-closed for
moderate values of $c$.

\section{Preliminaries}\label{sec:prelim}

We consider simple, undirected, and connected graphs $G = (V, E)$
with $n = |V|$ vertices and $m = |E|$ edges. For vertices $s, t ∈ V$
we write $d(s,t)$ for the \emph{distance} of $s$ and $t$, that is the
number of edges on a shortest path between $s$ and $t$. For
$i,j ∈ \mathbb{N}$, we write $[i]$ for the set $\{1,\dots,i\}$ and $[i,j]$ for
$\{i, \dots, j\}$. In a (unidirectional) \emph{breadth-first search}
(\emph{BFS}) from a vertex $s$, the graph is explored layer by layer
until the target vertex $t ∈ V$ is discovered. More formally, for a
vertex $v ∈ V$, the $i$-th \emph{BFS layer around $v$} (short:
\emph{layer}), $\ell_G(v,i)$, is the set of vertices that have
distance exactly $i$ from $v$. Thus, the BFS starts with
$\ell_G(s,0) = \{s\}$ and then iteratively computes $\ell_G(s,i)$
from $\ell_G(s,i-1)$ by iterating through the neighborhood of
$\ell_G(s,i-1)$ and ignoring vertices contained in earlier layers. We
call this the $i$-th \emph{exploration step} from $s$. We omit the
subscript $G$ from the above notation when it is clear from context.

In the \emph{bidirectional} BFS, layers are explored both from $s$
and $t$ until a common vertex is discovered. This means that the
algorithm maintains layers $\ell(s,i)$ of a \emph{forward search}
from $s$ and layers $\ell(t, j)$ of a \emph{backward search} from $t$
and iteratively performs further exploration steps in one of the
directions.
The decision about which search direction
to progress in each step is determined according to an \emph{alternation strategy}.
Note that we only allow the algorithm to switch between the search
directions after fully completed exploration steps.  
If the forward search performs $k$ exploration steps and the backward search
the remaining $d(s,t) - k$, then we say that the search \emph{meets} at layer
$k$.

In this paper, we analyze a particular alternation strategy called
the \emph{balanced} alternation strategy~\cite{borassi_kadabra_2016}.
This strategy greedily chooses to continue with an exploration step
in either the forward or backward direction, depending on which is
cheaper. Comparing the anticipated cost of the next exploration step
requires no asymptotic overhead, as it only requires summing the
degrees of vertices in the preceding layer. The following lemma gives
a running time guarantee for balanced BFS relative to any other
alternation strategy. This lets us consider arbitrary alternation
strategies in our mathematical analysis, while only costing a factor
of $d(s,t)$, which is typically at most logarithmic.

\begin{lemma}[{\cite[Theorem~3.2]{bff-espsf-22}}]\label{lemma_balanced_good}
    Let $G$ be a graph and $(s, t)$ a start--destination pair with
    distance $d(s,t)$. If there exists an alternation strategy such
    that the bidirectional BFS between $s$ and $t$ explores $f(n)$
    edges, then the balanced bidirectional search explores at most
    $d(s,t) \cdot f(n)$ edges.
\end{lemma}

The forward and backward search need to perform a total of $d(s,t)$
exploration steps. To ease the notation, we say that exploration step
$i$ (of the bidirectional search between $s$ and $t$) is either the
step of finding $ℓ(s,i)$ from $\ell(s,i-1)$ in the forward search or the
step of finding $ℓ(t, d(s,t) + 1 - i)$ from $ℓ(t, d(s,t) - i)$ in the
backward search. For example, exploration step 1 is the step in which
either the forward search finds the neighbors of $s$ or in which $s$
is discovered by the backwards search. Also, we identify the $i$-th
exploration step with its index $i$, i.e., $[d(s, t)]$ is the set of
all exploration steps. We often consider multiple consecutive
exploration steps together. For this, we define the interval
$[i, j] \subseteq [d(s, t)]$ to be a \emph{sequence} for
$i, j \in [d(s, t)]$. The \emph{exploration cost} of exploration step
$i$ from $s$ equals the number of visited edges with endpoints in
$ℓ(s,i-1))$, i.e., $c_{s}(i) = \sum_{v \in ℓ(s,i-1)}\deg(v)$. The
exploration cost for exploration step $i$ from $t$ is
$c_t(i) = \sum_{v∈ \ell(t,d(s, t) - i)} \deg(v)$. For a sequence
$[i, j]$ and $v \in \{s, t\}$, we define the cost
$c_v([i, j]) = \sum_{k \in [i, j]} c_{v}(k)$. Note that the notion of
exploration cost is an independent graph theoretic property and also
valid outside the context of a particular run of the bidirectional
BFS in which the considered layers are actually explored.

For a vertex pair $s, t$ we write $c_\mathrm{bi}(s,t)$ for the cost
of the bidirectional search with start $s$ and destination $t$. Also,
as we are interested in polynomial speedups, i.e., $\mathcal{O}(m^{1-ε})$
vs. $\mathcal{O}(m)$, we use $\tilde{\mathcal{O}}$-notation to suppress
poly-logarithmic factors.

\section{Performance Guarantees for Expanding Search Spaces}
\label{sec:expansion}

We now analyze the bidirectional BFS based on expansion properties.
In \cref{sec:expand-search-spac-basic-prop}, we introduce expansion,
including the concept of expansion overlap, state some basic
technical lemmas and give an overview of our results. In the
subsequent sections, we then prove our results for different cases of
the expansion overlap. Due to space limitations some proofs are in
the appendix.

\subsection{Expanding Search Spaces and Basic Properties}
\label{sec:expand-search-spac-basic-prop}

We define \emph{expansion} as the relative growth of the search space
between adjacent layers.  Let $[i, j]$ be a sequence of exploration
steps.  We say that $[i, j]$ is \emph{$b$-expanding from $s$} if for
every step $k \in [i, j)$ we have $c_s(k + 1) \ge b\cdot c_s(k)$.
Analogously, we define $[i, j]$ to be \emph{$b$-expanding from $t$} if
for every step $k \in (i, j]$ we have $c_t(k - 1) \ge b\cdot c_t(k)$.
Note that the different definitions for $s$ and $t$ are completely
symmetrical.
With this definition layed out, we investigate its relationship with
logarithmic distances.

\begin{lemma}
\label{lem_expansion_dist_small}
  Let $G=(V, E)$ be a graph and let $s, t ∈ V$ be vertices such that
  the sequence $[1, c \cdot d(s,t)]$ is $b$-expanding from $s$ for
  constants $b > 1$ and $c > 0$.  Then $d(s,t) ≤ \log_b(2m)/c$.
\end{lemma}
\begin{proof}
    The cost of discovering the layer with distance $c \cdot d(s,t)$ from $s$ is at least $b^{c \cdot d(s,t)}$.
    Thus we have $b^{c \cdot d(s,t)} ≤ 2m$, which can be rearranged to $c \cdot d(s,t) ≤ \log_b(2m)$.
\end{proof}

Note that this lemma uses $s$ and $t$ symmetrically and also applies to
expanding sequences from $t$. Together with \cref{lemma_balanced_good}, this
allows us to consider arbitrary alternation strategies that are convenient for
our proofs.
Next, we show that the total cost of a $b$-expanding sequence of exploration
steps is constant in the cost of the last step, which often simplifies
calculations.

\begin{restatable}{lemma}{lemFractionLastLayer}
  \label{lemma_fraction_last_layer}
    For $b>1$ let $f:ℕ \mapsto ℝ$ be a function with $f(i) ≥ b \cdot f(i-1)$ and $f(1) = c$ for some constant $c$.
    Then $f(n) / \sum_{i=1}^n f(i) ≥ \frac{b-1}{b}$.
\end{restatable}

We define four specific exploration steps depending on two constant parameters
$0<α<1$ and $b>1$. First, $\cheap_s(α)$ is the latest step such that
$c_s([1, \cheap_s(α)]) ≤ m^α$. Moreover, $\expan_s(b)$ is the latest step such
that the sequence $[1, \expan_s(b)]$ is $b$-expanding from $s$. Analogously, we
define $\cheap_t(α)$ and $\expan_t(b)$ to be the smallest exploration steps such
that $c_t([\cheap_t(α), d(s, t)]) ≤ m^α$ and $[\expan_t(b), d(s,t)]$ is
$b$-expanding from $t$, respectively. If $\expan_t(b) \le \expan_s(b)$, we say
that the sequence $[\expan_t(b), \expan_s(b)]$ is a \emph{$b$-expansion overlap}
of size $\expan_s(b) - \expan_t(b)+ 1$. See \cref{fig_seq_expl_steps} for a
visualization of these concepts. Note that the definition of $\expan_{s}$ (reps.
$\expan_t$) cannot be relaxed to only require expansion behind $\cheap_s$ (resp.
$\cheap_t$), as in that case an existing expansion overlap no longer implies
logarithmic distance between $s$ and $t$. This allows for the construction of
instances with linear running time (see \cref{lem:no_relax} in the appendix). To
simplify notation, we often omit the parameters $\alpha$ and $b$ as well as the
subscript $s$ and $t$ if they are clear from the context. Note that $\cheap_{s}$
or $\cheap_{t}$ is undefined if $c_{s}(1) > m^{α}$ or $c_{t}(d(s,t)) > m^{α}$.
Moreover, in some cases $\expan_v$ may be undefined for $v∈ \{s, t\}$, if the
first exploration step of the corresponding sequence is not $b$-expanding. Such
cases are not relevant in the remainder of this paper.

\begin{figure}[t]
    \centering
    \includegraphics[width = \textwidth]{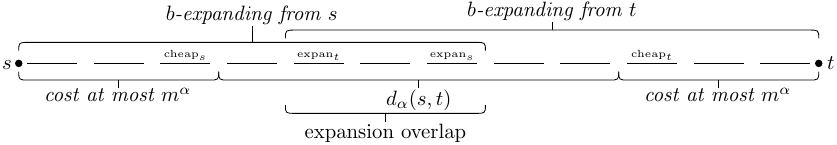}
    \caption{Visualization of $\cheap_v$, $\expan_v$ and related
    concepts.  Each line stands for an exploration step between $s$
    and $t$. Additionally, certain steps and sequences relevant for
    \cref{lem_sublinear_by_large_overlap} and
    \cref{lem_sublinear_with_cheap_start} are marked.}
    \label{fig_seq_expl_steps}
\end{figure}

\paragraph{Overview of our Results.}

Now we are ready to state our results.  Our first result
(\cref{lem_sublinear_by_large_overlap}) shows that for $b > 1$ we
obtain sublinear running time if the expansion overlap has size at
least $\Omega(\log m)$.  Note that this already motivates why the two
steps $\expan_s$ and $\expan_t$ and the resulting expansion overlap
are of interest.

The logarithmic expansion overlap required for the above result is of course a
rather strong requirement that does not apply in all cases where we expect
expanding search spaces to speed up bidirectional BFS. For instance, the
expansion overlap is at most the distance between $s$ and $t$, which might
already be too small. This motivates our second result
(\cref{lem_sublinear_with_cheap_start}), where we only require an expansion
overlap of sufficient relative length, as long as the maximum expansion is at
most a constant factor of the minimum expansion $b$. Additionally, we make use
of the fact that $\cheap_s$ and $\cheap_t$ can give us initial steps of the
search that are cheap. Formally, we define the \emph{($α$-)relevant distance} as
$d_α(s,t) = \cheap_t - \cheap_s - 1$ and require expansion overlap linear in
$d_\alpha(s, t)$, i.e., we obtain sublinear running time if the expansion
overlap is at least $c\cdot d_α(s,t)$ (see also \cref{fig_seq_expl_steps}) for
some constant $c$.

Finally, in our third result (\cref{thm_dichotomy}), we relax the
condition of \cref{lem_sublinear_with_cheap_start} further by
allowing expansion overlap that is sublinear in $d_α(s,t)$ or even
non-existent. The latter corresponds to non-positive expansion
overlap, when extending the above definition to the case
$\expan_t > \expan_s$. Specifically, we define $S_1 = \expan_s$,
$S_2 = \cheap_t - \expan_s - 1$, $T_1 = d(s, t) - \expan_t + 1$, and
$T_2 = \expan_t - \cheap_s - 1$ (see \cref{fig_seq_expl_steps_rho})
and give a bound for which values of
\begin{equation*}
    \rho = \frac {\max\{S_2, T_2\}} {\min\{S_1, T_1\}},
\end{equation*}
sublinear running time can be guaranteed. We write $ρ_{s,t}(α,b)$ if
these parameters are not clear from context. This bound is tight (see
\cref{lem_linear_example}), i.e., for all larger values of $\rho$ we
give instances with linear running time.

\subsection{Large Absolute Expansion Overlap}\label{sec:large-absol-expans}

We start by proving sublinear running time for a logarithmic
expansion overlap.

\begin{theorem}
\label{lem_sublinear_by_large_overlap}
    For parameter $b>1$ let $s, t ∈ V$ be a
    start--destination pair with a $b$-expansion overlap of size
    at least $c \log_b(m)$ for a constant $c>0$.  Then
    $c_\mathrm{bi}(s,t) ≤ 8\log_b(2m) \cdot \frac{b^2}{b-1} \cdot m^{1
    - c/2}$.
\end{theorem}
\begin{proof}
  We analyze bidirectional search when meeting in the middle
  $k_\mathrm{mid}$ (rounded either up or down) of the expansion
  overlap.  Using \cref{lemma_balanced_good} for an upper bound on
  the cost of the balanced bidirectional search under the assumed
  meeting point, we get
    \begin{equation*}
        c_\mathrm{bi}(s,t) ≤ d(s,t) \cdot \left(
            c_s([1, k_\mathrm{mid}])
            +
            c_t([k_\mathrm{mid}+1, d(s,t)])
            \right).
    \end{equation*}
    For an upper bound on $d(s,t)$, note that as there is an expansion
    overlap, $\expan_s ≥ d(s,t)/2$ or $\expan_t ≤ d(s,t) /2$.
    This means that $d(s,t)≤2\log_b(2m)$ by \cref{lem_expansion_dist_small}.
    Applying \cref{lemma_fraction_last_layer} we get
    \begin{align*}
        c_\mathrm{bi}(s,t) &≤ 2\log_b(2m) \cdot \frac{b}{b-1}
            \left(c_s(k_\mathrm{mid}) + c_t(k_\mathrm{mid}+1)\right),
        \intertext{which, assuming without loss of generality $c_s(k_\mathrm{mid}) ≥ c_t(k_\mathrm{mid}+1)$, gives us}
        &≤ 4\log_b(2m) \cdot \frac{b}{b-1} c_s(k_\mathrm{mid}).
    \end{align*}
    At least $\lfloor \frac12 c \log_b(m) \rfloor$ more $b$-expanding
    layers follow after $ℓ(s,k_\mathrm{mid})$.
    Counting the edges in these layers, we get
    \begin{equation*}
        c_s(k_\mathrm{mid}) \cdot b^{\lfloor \frac12 c \log_b(m) \rfloor} ≤ 2m,
    \end{equation*}
    which can be transformed to
    \begin{equation*}
        c_s(k_\mathrm{mid}) ≤ 2m \cdot b^{-\lfloor \frac12 c \log_b(m) \rfloor}.
    \end{equation*}
    Inserting this into the upper bound for $c_\mathrm{bi}(s,t)$, we get
    \begin{align*}
        c_\mathrm{bi}(s,t) &≤ 8\log_b(2m) \cdot \frac{b}{b-1} m \cdot
        b^{-\lfloor \frac12 c \log_b(m) \rfloor}\\
        &≤ 8\log_b(2m) \cdot \frac{b}{b-1} m \cdot b^{-\frac12 c
        \log_b(m) +1}\\
        &≤ 8\log_b(2m) \cdot \frac{b^2}{b-1} m \cdot b^{-\frac12 c \log_b(m)}\\
        &≤ 8\log_b(2m) \cdot \frac{b^2}{b-1} \cdot m^{1 - \frac12 c}.\qedhere
    \end{align*}
\end{proof}

\subsection{Large Relative Expansion Overlap}
\label{sec:large-relat-expans}

Note that \cref{lem_sublinear_by_large_overlap} cannot be applied if
the length of the expansion overlap is too small. We resolve this in
the next theorem, in which the required length of the expansion
overlap is only relative to $α$-relevant distance between $s$ and
$t$, i.e., the distance without the first few cheap steps around $s$ and
$t$. Additionally, we say that $b^+$ is the \emph{highest expansion
  between $s$ and $t$} if it is the smallest number, such that there
is no sequence of exploration steps that is more than $b^+$-expanding
from $s$ or $t$.

\begin{restatable}{theorem}{lemSublinearWithCheapStart}
\label{lem_sublinear_with_cheap_start}
  For parameters $0 \le α < 1$ and $b > 1$, let $s, t ∈ V$ be a
  start--destination pair with a $b$-expansion overlap of size at least
  $c \cdot d_α(s,t)$ for some constant $c>0$ and assume that $b^+ \ge b$ is
  the highest expansion between $s$ and $t$.  Then
  $c_\mathrm{bi}(s,t) ∈ \tilde{\mathcal{O}}\left(m^{1-ε}\right)$ for
  $ε = \frac{c(1-α)}{\log_{b}(b^+)+c} > 0$.
\end{restatable}

Note that \cref{lem_sublinear_with_cheap_start} does not require
$\expan_t > \cheap_s$ or $\expan_s < \cheap_t$, i.e., the expansion
overlap may intersect the cheap prefix and suffix.  Before extending
this result to an even wider regime, we want to briefly mention a
simple corollary of the theorem, in which we consider vertices with an
expansion overlap region and polynomial degree.

\begin{corollary}\label{cor_eps}
  For parameter $b>1$, let $s, t ∈ V$ be a start--destination pair with a
  $b$-expansion overlap of size at least $c \cdot d(s,t)$ for a constant
  $0<c≤1$. Further, assume that $\deg(t) ≤ \deg(s) ≤ m^δ$ for a constant
  $δ∈(0,1)$ and that $b^{+}$ is the highest expansion between $s$ and $t$. Then
  $c_\mathrm{bi}(s,t) ∈ \tilde{\mathcal{O}}\left(m^{1-\frac{c(1-δ)}{\log_{b}(b^+)+c}}\right)$.
\end{corollary}
This follows directly from \cref{lem_sublinear_with_cheap_start}, using
$\cheap_s(δ)$ and $\cheap_t(δ)$.

\subsection{Small or Non-Existent Expansion Overlap}
\label{sec:small-or-no-expan}

\cref{lem_sublinear_with_cheap_start} is already quite flexible, as it only
requires an expansion overlap with constant length relative to the distance
between $s$ and $t$, minus the lengths of a cheap prefix and suffix. In this
section, we weaken these conditions further, obtaining a tight criterion for
polynomially sublinear running time. In particular, we relax the length
requirement for the expansion overlap as far as possible. Intuitively, we
consider the case in which the cheap prefix and suffix cover almost all the
distance between start and destination. Then, the cost occurring between prefix
and suffix can be small enough to stay sublinear, regardless of whether there
still is an expansion overlap or not.

In the following we first examine the sublinear case, before
constructing a family of graphs with linear running time for the
other case and putting together the complete dichotomy in
\cref{thm_dichotomy}. We begin by proving an upper bound for the
length of low-cost sequences, such as $[1, \cheap_{s}]$ and
$[\cheap_t, d(s,t)]$.
\begin{restatable}{lemma}{lemExpSequenceShort}
\label{lem_exp_sequence_short}
    Let $v$ be a vertex with a $b$-expanding sequence $S$ starting at
    $v$ with cost $c_{v}(S) ≤ C$.  Then $|S| ≤ \log_b(C)+1$.
\end{restatable}

\begin{figure}[t]
    \centering
    \includegraphics[width = \textwidth]{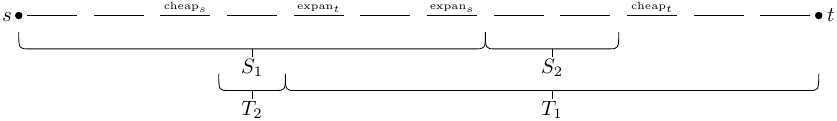}
    \caption{Visualization of exploration steps and (lengths of) sequences
      relevant for \cref{lem:old_claim:small_overlap,lem_sublinear_abc}}
    \label{fig_seq_expl_steps_rho}
\end{figure}

This statement is used in the following technical lemma that is needed to prove
sublinear running times in the case of small expansion overlap. Recall from
\cref{sec:expand-search-spac-basic-prop} that
$ρ_{s,t}(α,b) = \frac {\max\{S_2, T_2\} } {\min\{S_1, T_1\}}$; also
see \cref{fig_seq_expl_steps_rho}.

\begin{restatable}{lemma}{lemOldClaimSmallOverlap}
\label{lem:old_claim:small_overlap}
For parameters $0 \le α < 1$ and $b>1$, let $s, t ∈ V$ be a start--destination
pair and assume that $b^+$ is the highest expansion between $s$ and $t$ and
$\rho_{s,t}(α, b) < \frac{1-α}{1-α+α\log_b(b^+)}$. There are constants $c > 0$
and $k$ such that if the size of the $b$-expansion overlap is less than
$c \cdot \log_{b}(m) - k$, then there is a constant $x < 1$ such that
$c_{s}([1, \cheap_s + T_2]) ≤ 2^{1-α} \cdot m^x$ and
$c_{t}([\cheap_{t} - S_{2}, d(s,t)]) ≤ 2^{1-α} \cdot m^x$.
\end{restatable}
This lets us prove the sublinear upper bound.

\begin{lemma}\label{lem_sublinear_abc}
    For parameters $0 \le α < 1$ and $b>1$, let $s, t ∈ V$ be a start--destination
    pair and assume that $b^+$ is the highest expansion between $s$
    and $t$.
    If $\rho_{s,t}(α, b) < \frac{1-α}{1-α+α\log_b(b^+)}$, then $c_\mathrm{bi}(s,t) ∈
    \tilde{\mathcal{O}}\left(m^{1-ε}\right)$ for a constant $ε>0$.
\end{lemma}
\begin{proof}
  We make a case distinction on the size of the expansion overlap.
  First, we note that by \cref{lem_sublinear_by_large_overlap}, the
  bidirectional search has sublinear running time if the
  $b$-expansion overlap has size $\Omega(\log m)$.

  For the other case, we consider an expansion overlap of size less
  than $c \cdot \log_{b}(m) -k$, for constants $c > 0$ and $k$ as in
  \cref{lem:old_claim:small_overlap}. We analyze the cost of doing
  $\cheap_s + T_2$ exploration steps in the forward search and,
  symmetrically, $(d(s,t) - \cheap_t + 1) + S_2$ in the backward
  search. By \cref{lem:old_claim:small_overlap}, these sequences have
  sublinear cost, as there is an $x<1$ such that
  $c_{s}([1, \cheap_s + T_2]) ≤ 2^{1-α} \cdot m^x$ and
  $c_t([\cheap_t-S_2, d(s,t)]) ≤ 2^{1-α} \cdot m^x$. Without loss of
  generality, assume $S_{1} \ge T_{1}$. We again consider two cases.

  If $\expan_{s} > \cheap_{s} + T_{2}$, the
  expansion-overlap is reached after the considered sequences of
  exploration steps. As the cost of these sequences is in $O(m^x)$, we
  have $\cheap_{s}(α) + T_{2} \le \cheap_{s}(x)$,
  $\cheap_{t}(α) - S_{2} \ge \cheap_{t}(x)$ and the size of the
  expansion-overlap is at least the $x$-relevant distance between $s$
  and $t$, which gives sublinear running time according to
  \cref{lem_sublinear_with_cheap_start}. In the other case we have
  $\expan_{s} \le \cheap_{s} + T_{2}$. Thus, the considered sequences
  of exploration steps overlap, as
  $\cheap_{s} + T_{2} + 1 \ge \cheap_{t} - S_{2}$. By considering the
  cost under an assumed meeting point at the end of one of the
  considered sequences, sublinear running time for the entire
  bidirectional search follows.
\end{proof}


The following lemma covers the other side of the dichotomy, by
proving a linear lower bound on the running time for the case where
the conditions on $\rho$ in \cref{lem_sublinear_abc} are not met. The
proof can be found in \cref{sec:omitted-proofs}, but the rough idea
is the following. We construct symmetric trees of depth $d$ around
$s$ and $t$. The trees are $b$-expanding for $(1-ρ)d$ steps and
$b^{+}$-expanding for subsequent $ρd$ steps and are connected at
their deepest layers.

\begin{restatable}{lemma}{lemLinearExample}
  \label{lem_linear_example}
    For any choice of the parameters $0<α<1$, $b^+>b>1$, $ρ_{s, }(α,b) ≥
    \frac{1-α}{1-α+α\log_{b}(b^+)}$ there is an infinite family of
    graphs with two designated vertices $s$ and $t$, such that in the
    limit $\cheap_s(α)$, $\cheap_t(α)$, $\expan_s(b)$, and
    $\expan_t(b)$ fit these parameters, $b^+$ is the highest
    expansion between $s$ and $t$ and $c_\mathrm{bi}(s,t) ∈ Θ(m)$.
\end{restatable}

This lets us state a complete characterization of the worst case running time
of bidirectional BFS depending on $ρ_{s,t}(α,b)$.  It follows directly from
\cref{lem_sublinear_abc} and \cref{lem_linear_example}.

\begin{theorem}\label{thm_dichotomy}
  Let an instance $(G, s, t)$ be a graph with two designated vertices, let
  $b^{+}$ be the highest expansion between $s$ and $t$ and let $0<α<1$ and $b>1$
  be parameters. For a family of instances we have
  $c_{\mathrm{bi}}(s,t) ∈ \mathcal{O}(m^{1-ε})$ for some constant $ε>0$ if
  $ρ_{s,t}(α,b) < \frac{1-α}{1-α+α\log_b(b^+)}$ and
  $c_{\mathrm{bi}}(s,t) ∈ Θ(m)$ otherwise.
\end{theorem}

\section{Evaluation}\label{sec:eval}

We conduct experiments to evaluate how well our concept of expansion
captures the practical performance observed on real-world networks.
For this, we use a collection of 3006 networks selected from Network
Repository~\cite{blasius_thomas_2022_6586185,networkrepository}. The
data-set was obtained by selecting all networks with at most
\SI{1}{M} edges and comprises networks from a wide range of domains
such as social-, biological-, and infrastructure-networks. Each of
these networks was reduced to its largest connected component and
multi-edges, self-loops, edge directions and weights were ignored.
Finally, only one copy of isomorphic graphs was kept. The networks
have a mean size of $12386$ vertices (median $522.5$) and are mostly
sparse with a median average degree of $5.6$. More statistics on the
networks can be found in the appendix.

\subsection{Setup \& Results}\label{sec:eval:results}
For each graph, we randomly sample 250 start--destination pairs
$s, t$.  We measure the cost of the bidirectional
search as the sum of the degrees of explored vertices.  For each graph
we can then compute the average cost $\hat{c}$ of the sampled pairs.
Then, assuming that the cost behaves asymptotically as $\hat c = m^x$
for some constant $x$, we can compute the \emph{estimated exponent} as
$x = \log_m \hat{c}$.

We focus our evaluation on the conditions in
\cref{lem_sublinear_with_cheap_start} and \cref{thm_dichotomy}. For this, we
compute $\expan_s(b)$, $\expan_t(b)$, $\cheap_s(α)$, and $\cheap_t(α)$ for each
sampled vertex pair for all values of $\alpha$, by implicitly calculating the
values of $\alpha$ corresponding to cheap sequences of different length.

\begin{figure}[t]
    \centering
    \begin{subfigure}[b]{0.45\textwidth}
        \includegraphics[width=\textwidth]{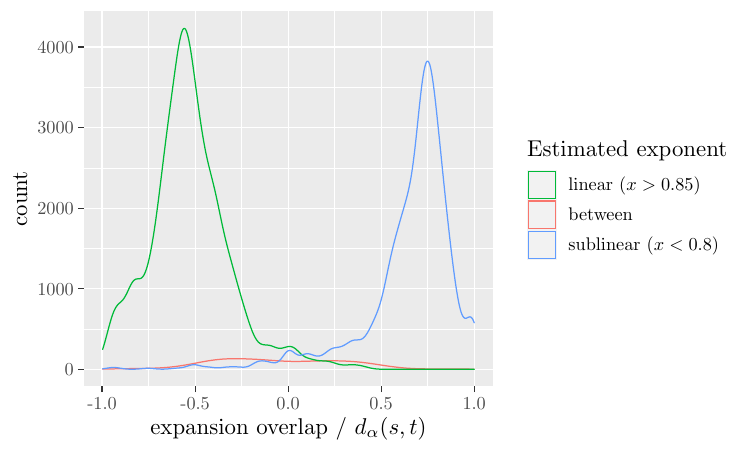}
        \caption{Distribution of parameter $c$ of
        \cref{lem_sublinear_with_cheap_start} for $b=2$ and $\alpha =
        0.1$ for graphs with different estimated exponents.}
        \label{fig_lemma_6_density}
    \end{subfigure}
    ~ 
    \begin{subfigure}[b]{0.45\textwidth}
        \includegraphics[width=\textwidth]{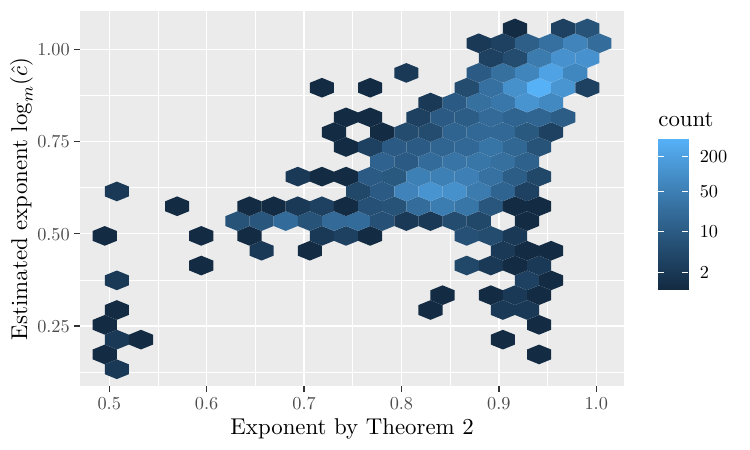}
        \caption{Relationship between estimated exponent and asymptotic exponent
          predicted by \cref{lem_sublinear_with_cheap_start} for $b=2$.}
        \label{fig_lemma_6_exponent}
    \end{subfigure}
    \caption{Empirical validation of \cref{lem_sublinear_with_cheap_start}.}\label{fig_experiments_lem6}
\end{figure}

By \cref{lem_sublinear_with_cheap_start} a vertex pair has
asymptotically sublinear running time, if the length of the expansion
overlap is a constant fraction of the relevant distance $d_α(s, t)$.
We therefore computed this fraction for every pair and
then averaged over all sampled pairs of a graph.  Note that for any
graph of fixed size, there is a value of $α$, such that
$\cheap_s(α) ≥ \cheap_t(α)$.  We therefore set $α≤0.1$ in order to not
exploit the asymptotic nature of the result.  Also we set the minimum
base of the expansion $b$ to $2$.  Outside of extreme ranges, the
exact choice of these parameters makes only little difference, as
discussed in more detail in \cref{sec:appendix-plots}.
\cref{fig_lemma_6_density} shows the distribution of the relative
length of the expansion overlap for different values of the estimated
exponent.  It separates the graphs into three categories; graphs with
estimated exponent $x > 0.85$ ((almost) linear), with $x < 0.8$
(sublinear) and the graphs in between.  We note that the exact choice
of these break points makes little difference; more detailed plots can
be found in \cref{sec:appendix-plots}.

Note that \cref{lem_sublinear_with_cheap_start} states not only
sublinear running time but actually gives the exponent as
$1 - \frac{c(1-α)}{2(\log_{b}(b^+)+c)}$.  \cref{fig_lemma_6_exponent}
shows the relationship between this exponent (averaged over the $(s, t)$-pairs)
and the estimated exponent.  For each sampled pair of vertices we
chose $\alpha$ optimally to minimize the exponent.  This is valid even
for individual instances of fixed size, because even while higher
values of $\alpha$ increase the fraction of the distance that is
included in the cheap prefix and suffix, this increases the predicted
exponent.  

Finally, \cref{thm_dichotomy} proves sublinear running time if
$\rho_{s,t}(α,b) ≤ \frac{1-α} {1-α+α \log_b(b^+)}$.  To evaluate how well
real-world networks fit this criterion, we computed $\rho_{s,t}(α,b)$
for each sampled pair $(s,t)$ as well as the upper bound
$\rho_\mathrm{max} := \frac{1-α} {1-α+α \log_b(b^+)}$.  Again, choosing
large values for $α$ does not exploit the asymptotic nature of the
statement, as $\rho_\mathrm{max}$ tends to $0$ for large values of
$α$.  For each vertex pair, we therefore picked the optimal value of
$α$, minimizing $\rho_\mathrm{max} - \rho_{s,t}(α,b)$ and recorded the
average over all pairs for each graph.  \cref{fig_exp_thm10} shows the
difference between $1/(1+\rho_{s,t}(α,b))$ and
$1/(1+\rho_\mathrm{max})$.  This limits the range of these values to
$[0,1]$ and is like dividing $S_2$ by $S_1 + S_2$ instead of
$S_2$ by $S_1$ in the definition of $\rho$.

\begin{figure}[t]
    \centering
    \begin{subfigure}[b]{0.45\textwidth}
        \includegraphics[width=\textwidth]{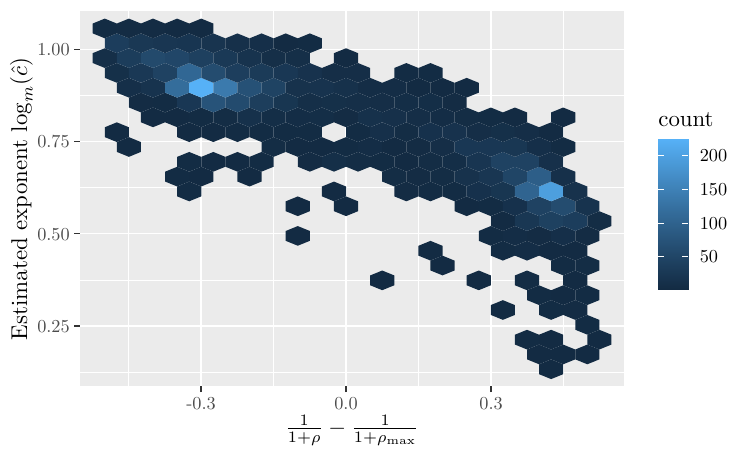}
    \end{subfigure}
    ~ 
    \begin{subfigure}[b]{0.45\textwidth}
        \includegraphics[width=\textwidth]{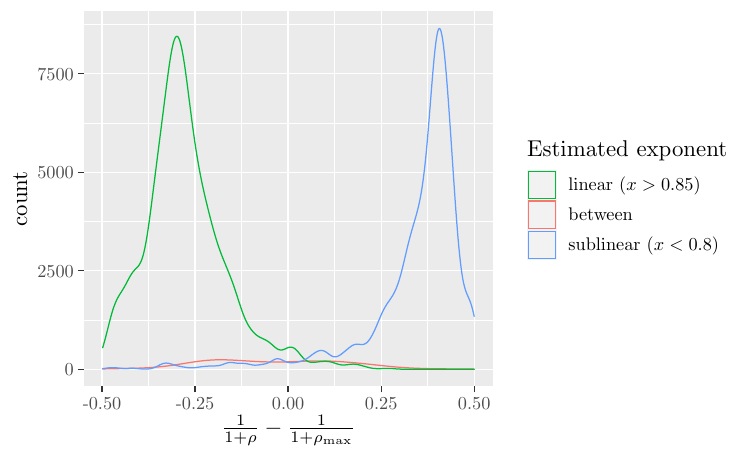}
    \end{subfigure}
    \caption{Relationship between estimated exponent and $δ_ρ = 1/(1+\rho_{s,t}(α,b) - 1/(1+\rho_\mathrm{max})$ for $b=2$. \cref{thm_dichotomy} predicts sublinear running time for all points with $δ_ρ > 0$.}\label{fig_exp_thm10}
\end{figure}

\subsection{Discussion}\label{sec:eval:discussion}

Both \cref{fig_exp_thm10} and \cref{fig_lemma_6_density} show that
our notion of expansion not only covers some real networks, but
actually gives a good separation between networks where the
bidirectional BFS performs well and those where it requires (close
to) linear running time. With few exceptions, exactly those graphs
that seem to have sublinear running time satisfy our conditions for
asymptotically sublinear running time. Furthermore, although the
exponent stated in \cref{lem_sublinear_with_cheap_start} only gives an
asymptotic worst-case guarantee, \cref{fig_lemma_6_exponent} clearly
shows that the estimated exponent of the running time is strongly
correlated with the exponent given in the theorem.

%
%
%
\bibliographystyle{splncs04}
\bibliography{bib}

\clearpage
\appendix

\section{Appendix}\label{sec:apx}

\subsection{Omitted proofs}\label{sec:omitted-proofs}

\lemFractionLastLayer*
\begin{proof}
    We have $f(i-1) ≤ f(i)/b$ and so we get
    \begin{equation*}
         \frac{f(n)}{\sum_{i=1}^n f(i)}
         ≥ \frac{f(n)}{\sum_{i=0}^n f(n)/b^i}
         = \frac{1}{\sum_{i=0}^n 1/b^i}
         = \frac{1-1/b}{1-1/b^{n+1}}
         = \frac{b-1}{b-1/b^{n+2}}
         ≥ \frac{b-1}{b}.\qedhere
    \end{equation*}
\end{proof}

\begin{remark}\label{lem:no_relax}
  If the definition of $\expan_{s}$ (respectively $\expan_{t}$) is
  relaxed to only require $b$-expansion in the sequence of
  exploration steps $[\cheap_s(α), \expan_{s}]$ (respectively
  $[\expan_t, \cheap_t(α)]$), then we can construct instances with
  logarithmic expansion overlap on which the cost of the
  bidirectional search is linear.
\end{remark}
\begin{proof}
  Assume $\frac12 \le α < 1$ and $b>1$. We construct an instance as sketched below.\\
  \begin{center}
    \includegraphics[]{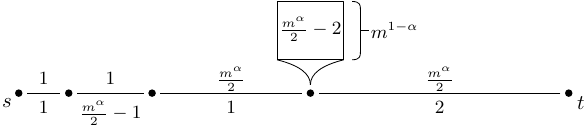}
  \end{center}
  We connect $s$ to the rest of the graph via one layer with cost $1$, a second
  layer with cost $\frac{m^α}{2}-1$ and another $\frac{m^α}{2}$ layers of cost
  $1$ and $t$ via $\frac{m^α}{2}$ layers of cost $2$. Behind these cheap regions
  of cost $m^α$, we append a logarithmic number of $b$-expanding layers followed
  by $Θ(m^{1-α})$ layers of cost $\frac{m^α}{2}-2$. This way, assuming the relaxed
  definition of $\expan_{s}$ and $\expan_{t}$, there is an expansion overlap of
  logarithmic size. However, the balanced alternation strategy will only perform
  one step in the forward direction and instead explore the (individually)
  cheaper layers of cost $2$ and $\frac{m^α}{2}-2$. This leads to linear overall
  cost.
\end{proof}

\lemSublinearWithCheapStart*
\begin{proof}
  As there is an expansion overlap of length $c \cdot d_α(s,t)$, we
  can apply \cref{lem_sublinear_by_large_overlap} if $d_α(s,t)$ is
  large enough.  Assume that $d_α(s,t) ≥ a \log_b(m)$ for some
  constant $a$ to be determined later.  Then, by
  \cref{lem_sublinear_by_large_overlap} we immediately get a
  sublinear upper bound
    \begin{equation*}
        c_\mathrm{bi}(s,t) ≤ 8\log_{b}(2m) \cdot \frac{b^2}{b-1} \cdot m^{1 - ca/2},
    \end{equation*}
    if we can choose $a$ suitably.

    Otherwise, we have $d_α(s,t) < a \log_b m$, in which case we can find an
    upper bound for $c_{\mathrm{bi}}(s,t)$ by considering the cost for an
    assumed meeting point in the middle between $\cheap_{s}$ and $\cheap_{t}$,
    i.e., after $\cheap_s + \lceil d_α(s,t) / 2 \rceil$ steps of the forward
    search and $(d(s,t) - \cheap_t + 1) + \lfloor d_{α}(s,t)/2 \rfloor$ steps of
    the backward search. Via \cref{lemma_balanced_good} we thus get
    \begin{align*}
      \begin{split}
        c_\mathrm{bi}(s,t) &≤ d(s,t) \cdot
            \left(c_s([1, \cheap_s + \lceil d_α(s,t)/2 \rceil])
            +\right.\\
            &\qquad \left.c_t([\cheap_t - \lfloor d_{α}(s,t)/2 \rfloor, d(s,t)])\right).
      \end{split}\\
      \intertext{
          Pessimistically assuming $c_s(i+1) = c_s(i)
          \cdot b^+$ for $i≥\cheap_s$, we can use
          \cref{lemma_fraction_last_layer} to obtain
      }
      \begin{split}
        &≤ d(s,t) \cdot
            \left(\frac{b^+}{b^+-1} c_s(\cheap_s + d_α(s,t)/2+1)
            +\right.\\
            &\qquad \left.\frac{b^+}{b^+-1} c_t(\cheap_t - d_{α}(s,t)/2)\right)
        \end{split}\\
        \begin{split}
        &≤ d(s,t) \cdot
            \left(\frac{{b^+}^{2}}{b^+-1} c_s(\cheap_s)\cdot {b^+}^{d_α(s,t)/2}
            +\right.\\
            &\qquad \left.\frac{b^+}{b^+-1} c_t(\cheap_t) \cdot { b^+ }^{d_{α}(s,t)/2}\right),
        \end{split}\\
        \intertext{where we apply $c_s(\cheap_s) ≤ c_s([1,\cheap_s]) ≤
        m^α$ and symmetrically $c_t(\cheap_t) ≤ c_t([\cheap_t,d(s,t)]) ≤ m^α$ and $d_α(s,t) < a \log_b m$ to get}
        &≤ d(s,t) \cdot
            \left(2 \cdot \frac{{b^+}^{2}}{b^+-1} m^α\cdot {b^+}^{a \log_b m / 2}\right)\\
        &∈ O\left(
            d(s,t) \cdot
            m^{α+a\log_b(b^+) / 2}
        \right).
    \end{align*}
    In order to find the optimal choice of $a$, we set the two
    exponents from the case distinction to be equal
    \begin{align*}
        1-ca/2 &= α + a \log_{b}(b^+) / 2\\
        1-α &= a \log_{b}(b^+)/2 + ca/2\\
        a &= \frac{2(1-α)}{\log_{b}(b^+) + c}.
    \end{align*}
    Thus we have
    $c_\mathrm{bi}(s,t) ∈ \mathcal{O}\left(d(s,t) \cdot m^{1-\frac{c(1-α)}{\log_{b}(b^+)+c}}\right) \subseteq \tilde{\mathcal{O}}\left( m^{1-\frac{c(1-α)}{\log_{b}(b^+)+c}} \right)$.
\end{proof}

\lemExpSequenceShort*
\begin{proof}
  We have
  \begin{equation*}
    C ≥ c_{v}(S)
    = \sum_{i=1}^{|S|} c_{v}(i)
    ≥ \sum_{i=0}^{|S|-1} b^i
    ≥ b^{|S|-1}
  \end{equation*}
  and thus get $|S|≤\log_b(C)+1$.
\end{proof}

\lemOldClaimSmallOverlap*
\begin{proof}
  We write $d_{\mathrm{overlap}}$ for the size of the expansion overlap.
  With \cref{fig_seq_expl_steps_rho} as reference, it is easy to
  verify that $d_{\mathrm{overlap}} = S_{1} - T_{2} - \cheap_{s}$.
  Note that this also holds in the case of
  $d_{\mathrm{overlap}} \le 0$. Without loss of generality assume
  $S_1 ≥ T_1$. Together with the definition of $ρ$ this implies
  $S_1 ρ ≥ T_1 ρ ≥ \max\{S_2, T_{2}\}$. We use $T_{2} \le S_{1} ρ$
  and $S_{1} \ge \frac{\max\{S_{2}, T_{2}\}}{ρ}$ to obtain
  \begin{align*}
    d_{\mathrm{overlap}} &= S_{1} - T_{2} - \cheap_{s}\\
                           &≥ S_{1} (1-ρ) - \cheap_{s}\\
                           &≥ \frac{1-ρ}{ρ} \max\{S_{2}, T_{2}\} - \cheap_{s},
  \end{align*}
  which we rephrase as
  \begin{align}\label{eq:1}
    \max\{S_{2}, T_{2}\} ≤ \frac{ρ}{1-ρ} (d_{\mathrm{overlap}} + \cheap_{s}).
  \end{align}
  This means that a small expansion overlap also implies small $S_2$ and $T_2$.

  We use this to derive a suitable upper bound on the expansion overlap
  that gives an upper bound on $T_{2}$ and $S_{2}$ for which the desired
  sublinear cost follows. As no exploration step is more than
  $b^{+}$-expanding, we have
  \begin{align}\label{eq:thm8:eq1}
    c_s(\cheap_{s} + T_{2}) ≤ m^{α} \cdot {b^+}^{T_2}
  \end{align}
  if we pessimistically assume maximum expansion in every step. Note
  that under the pessimistic assumption of $b^{+}$-expansion, this is
  at least a constant fraction of the cost in
  $c_{s}([1, \cheap_s + T_2])$ and by symmetry also
  $c_t([\cheap_t-S_2, d(s,t)])$. We use \cref{eq:thm8:eq1} and derive
  \begin{align*}
    c_s(\cheap_s + T_2) &≤ m^α \cdot m^{T_2 \cdot \log_m(b^+)}\\
                        &= m^α \cdot m^{T_2 \cdot \log_m(b^+) - (1-α) \cdot \log_m2} \cdot m^{(1-α)\log_m2}\\
                        &= 2^{1-α} \cdot m^{α + T_2 \cdot \log_m(b^+) - (1-α) \cdot \log_m2}.
  \end{align*}
  Clearly, this is sublinear if the exponent of $m$ is smaller than
  $1$. Investigating this, we get
  \begin{align*}
    α + T_{2} \cdot \log_{m}(b^{+}) - (1-α) \log_{m}2 &< 1\\
    T_{2} \log_{m}(b^+) &< 1-α + (1-α) \log_{m}2\\
    T_{2} &< \frac{(1-α)(1+\log_{m} 2)}{\log_{m}(b^+)}\\
    T_{2} &< (1-α) \log_{b^{+}}(2m).
  \end{align*}
  Using $T_{2} ≤ \frac{ρ}{1-ρ} (d_{\mathrm{overlap}} + \cheap_{s})$ from above, we continue with a stricter inequality
  \begin{align*}
    \frac{ρ}{1-ρ} (d_{\mathrm{overlap}} + \cheap_{s}) &< (1-α) \log_{b^{+}}(2m)\\
    d_{\mathrm{overlap}} + \cheap_{s} &< \frac{(1-α)(1-ρ)}{ρ} \log_{b^{+}}(2m)\\
    \intertext{
  and use $α\log_b(2m)+1$ as an upper bound from \cref{lem_exp_sequence_short} for
  $\cheap_s$ to derive a sufficient upper bound on
  $d_{\mathrm{overlap}}$ as
    }
    d_{\mathrm{overlap}} &< \frac{(1-α)(1-ρ)}{ρ} \log_{b^{+}}(2m) - α \log_{b}(2m) - 1\\
    d_{\mathrm{overlap}} &< \left(\frac{(1-α)(1-ρ)}{ρ\log_{b}{b^{+}}} - α \right) \log_{b}(2m) - 1.
  \end{align*}
  Relying on the initial assumption on $\rho$, we verify that the factor before the logarithm is a positive constant
  \begin{align*}
    \frac{(1-α)(1-ρ)}{ρ\log_{b}{b^{+}}} - α &> 0\\
    \frac{1-α}{ρ \log_{b}(b^+)}
    - \frac{ρ(1-α)}{ρ \log_b(b^+)}
    - \frac{αρ\log_b(b^+)}{ρ\log_b(b^+)} &> 0\\
    \frac{
    1-α - ρ(1-α) - αρ\log_b(b^+)
    }{
    ρ \log_{b}(b^+)
    }
                                            &> 0\\
    1 - α &> ρ(1-α) + ρ α \log_b(b^+)\\
    ρ &< \frac{1-α}{1-α + α \log_b(b^+)}.
  \end{align*}
  Thus the condition under which the considered sequences of exploration steps has sublinear cost can be expressed as $d_{\mathrm{overlap}} < c \cdot \log_{b}(m) - k$ for positive constants $c$ and $k$.
\end{proof}

\lemLinearExample*
\begin{proof}
  We construct such an instances by taking two isomorphic trees $T_s$ and $T_t$
  with roots $s$ and $t$ and connecting their deepest leaves with a matching.
  Let $d_1 + d_2 = d$ be the depth of these trees, with $d ∈ Θ(\log m)$. The
  number of branches at each layer is chosen so that $s$ and $t$ are
  $b$-expanding for $d_1$ steps and then $b^+$-expanding for the remaining $d_2$
  steps.

  In the following, we verify that this construction satisfies our requirements
  asymptotically. First, note that ignoring constant factors for $i≤d_1$ we have
  $c_s(i) ∈ Θ\left({b}^i\right)$ and for $d_1<i≤d_1+d_2$ we have
  $c_s(i) ∈ Θ\left({b}^{d_1} \cdot {b^+}^i\right)$. This means that
  $c_s(d_1+d_2)$ is linear in the total cost of $\sum_{i=1}^{d_1+d_2} c_s(i)$
  and therefore also linear in the total number of edges. This means that the
  most expensive layers are in the middle, just before the two trees meet. As
  there are no shortcuts, both the bidirectional search and the unidirectional
  search have to explore at least one of the two most expensive layers,
  resulting in linear running time overall.

  We now determine a suitable choice of $d_{1}$ and $d_{2}$.
  The idea is that from $s$, we want $d$ (at least) $b$-expanding layers,
  followed by $d_2$ no longer expanding layers, and $d_1$ layers with low cost
  $m^{α}$. In other words, the goal is to have $\expan_s(α,b) = d_1 + d_2$ and
  $\cheap_s(α) = d_1$, and analogously $\expan_t(α,b) = d_1+d_2+1$ and
  $\cheap_t = d_1 + 2 d_2 + 1$. In order to make the construction fit to the
  definitions, we need to ensure that $c_s(\cheap_s) ≤ m^α$ and the length
  $S_2 = T_{2}$ is sufficiently small compared to $S_1 = T_{1}$, that is,
  $d_2 ≤ ρ d$. As the construction is symmetrical, it suffices to check this for
  $s$.

  For the cost of region the first $d_1$ steps we have
  $c_s(\cheap_s) ∈ Θ\left(c_s(d_1)\right) = Θ\left({b}^{d_1}\right)$. Also we
  have
  $m ∈ Θ\left(c_s( d_1+d_2)\right) = Θ\left(b^{d_1}\cdot {b^+}^{d_2}\right) = Θ\left({b}^{d_1 + d_2 \log_{b}(b^+)}\right)$.
  This means that if the exponent of $b^{d_1}$ is at most the exponent of
  $b^{α\left(d_1 + d_2 \log_{b}(b^+)\right)}$, we get
  $c_{s}(\cheap_{s}) \in Θ(m^{α})$. In order to also get
  $c_{s}(\cheap_{s}) ≤ m^{α}$, we additionally need $b^{d_1+d_2 \log_b(b^+)}$ to
  be a sufficiently small fraction of $m$. We ensure this by appending a
  sufficiently long (linear in the size of the construction) path to some vertex
  in layer $d$. This does not asymptotically change the cost of any layer, but
  makes the number of edges in layer $d$ an arbitrarily small fraction of the
  total number of edges. It remains to compare the exponents of $b^{d_1}$ and
  $b^{α\left(d_1 + d_2 \log_{b}(b^+)\right)}$, which let's us derive the
  following requirement on $α$, $b$, and $b^+$ subject to $d_1$ and $d_2$.
    \begin{align*}
        d_1 &≤ α d_1 + α \log_{b}(b^+) d_2\\
        d_1 (1-α) &≤ α \log_{b}(b^+) d_2\\
        \frac{d_2}{d_1} &≥ \frac{1-α}{α \log_{b}(b^+)}.
    \end{align*}
    We now set $d_2 = ρ d$ and $d_1 = (1-ρ) d$.
    Then, this translates to
    \begin{align*}
        \frac{d_2}{d_1}
        = \frac{ρ}{1-ρ}
        ≥ \frac{1-α}{α \log_{b}(b^+)}.
    \end{align*}
    It is easy to verify that $\frac{x}{1-x} ≥ y$ and $x ≥ \frac{y}{1+y}$ are
    equivalent.  This means we get
    \begin{align*}
        ρ ≥ \frac{
            (1-α)/(α \log_{b}(b^+))
        }{
            1+ (1-α) / (α \log_{b}(b^+))
        }
         = \frac{1-α}{1 - α + α \log_{b}(b^+)},
    \end{align*}
    which matches exactly the claimed requirement for $ρ$.  This means
    that for any set of parameters with $ρ ≥
    \frac{1-α}{1-α+α\log_b(b^+)}$, we can construct an instance in
    which $s$ and $t$ fulfil these parameters and in which the
    bidirectional search between $s$ and $t$ has linear cost.
\end{proof}

\subsection{Supplementary remarks on the experiments}
\label{sec:appendix-plots}

\begin{figure}[t]
    \centering
    \begin{subfigure}[b]{0.45\textwidth}
        \includegraphics[width=\textwidth]{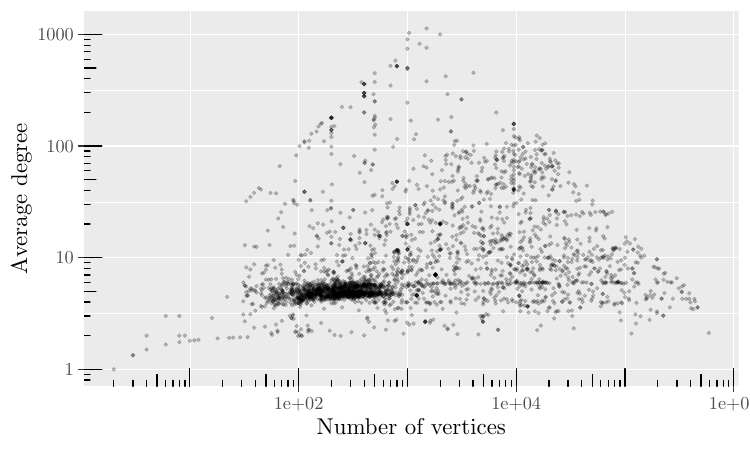}
    \end{subfigure}
    ~ 
    \begin{subfigure}[b]{0.45\textwidth}
        \includegraphics[width=\textwidth]{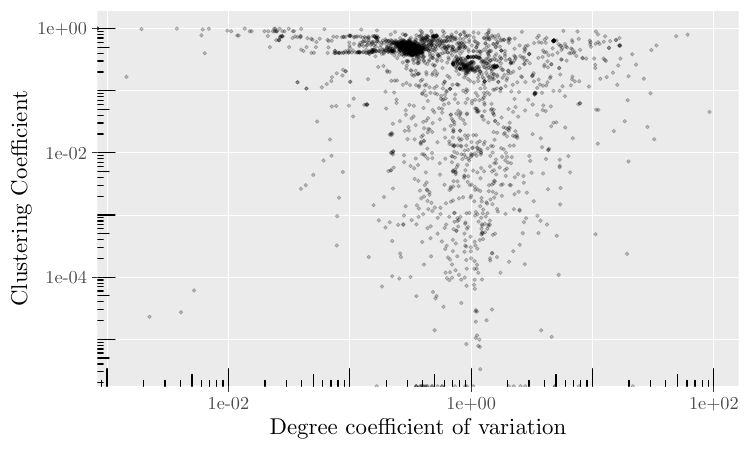}
    \end{subfigure}
  \caption{Overview of network size and average degree as well as
    clustering coefficient and heterogeneity of the degree
    distribution (measured via its coefficient of variation).}
  \label{fig:network_stats}
\end{figure}

First, we include some more statistics on the dataset of chosen real
world networks. The networks have a mean size of $12386$ vertices
(median $522.5$) with a mean average degree of $21.7$ (median: $5.6$)
and median clustering coefficient of $0.45$. See also
\cref{fig:network_stats} for a visual overview of these properties.
The average length of shortest paths varies a lot between the
different networks depending on their structure. We found a mean
average shortest path length of $28.5$ with a median of $5.03$.

In \cref{sec:eval:results} we briefly mentioned that the exact choice of $α$ and $b$ as well as the boundaries for the classification into linear and sublinear running time have little impact on the qualitative nature of our experimental results.
In order to explain this, we consider how choosing different values for $α$ and $b$ affects the results of our experiments.

\begin{figure}
    \centering
    \includegraphics[width = \textwidth]{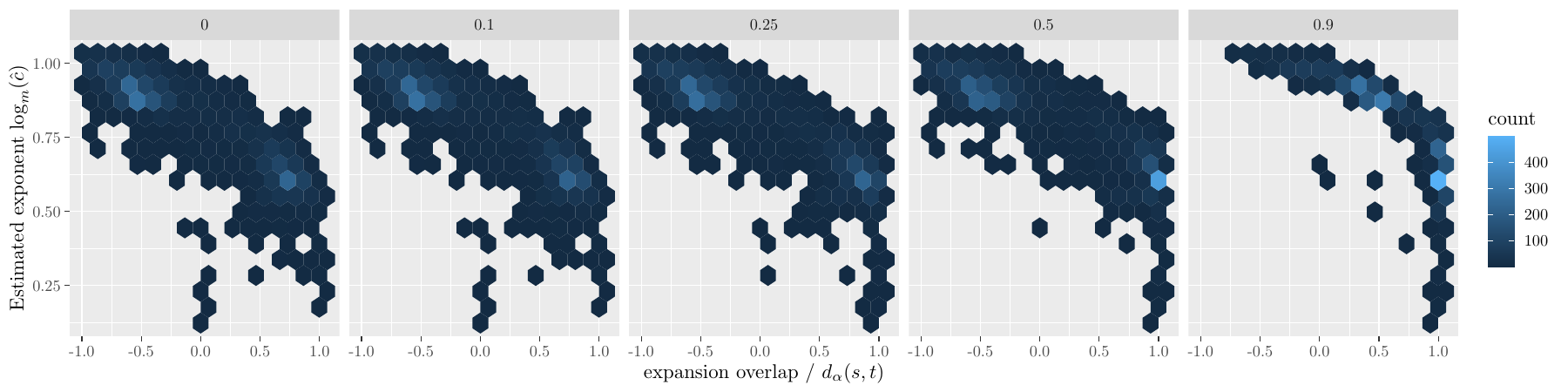}
    \caption{Estimated exponent for parameter $c$ of
      \cref{lem_sublinear_with_cheap_start} under $b=2$ and different values of
      $\alpha$.}
    \label{fig_different_alphas}
\end{figure}

In \cref{fig_different_alphas}, we show how the distribution underlying \cref{fig_lemma_6_density} changes if  we vary $\alpha$.
It can clearly be seen that the overall relationship between how easily the expansion overlap fulfils the condition of \cref{lem_sublinear_with_cheap_start} and the estimated exponent remains intact for all values, even though this inflates the length of the cheap prefix $\cheap_s$ and suffix $d(s,t) - \cheap_t$.
However, this effect only really begins to let apparently linear instances match the condition for sublinear running time for very large values of $α$ like $0.9$.
Contrary for all other depicted values of $α$ the distributions barely change and also the difference between not using the cheap regions at all by setting $α = 0$ and using them is not large.

It can also be seen that the arbitrary choice of when to classify instances as linear or sublinear based on their estimated exponent does not matter much.
For any of the smaller values of $α$, the distribution of the graphs is well separated into two clusters of high density, for which the estimated exponent is high, where the condition for \cref{lem_sublinear_with_cheap_start} is not fulfilled, and significantly lower otherwise.

\begin{figure}
    \centering
    \includegraphics[width = \textwidth]{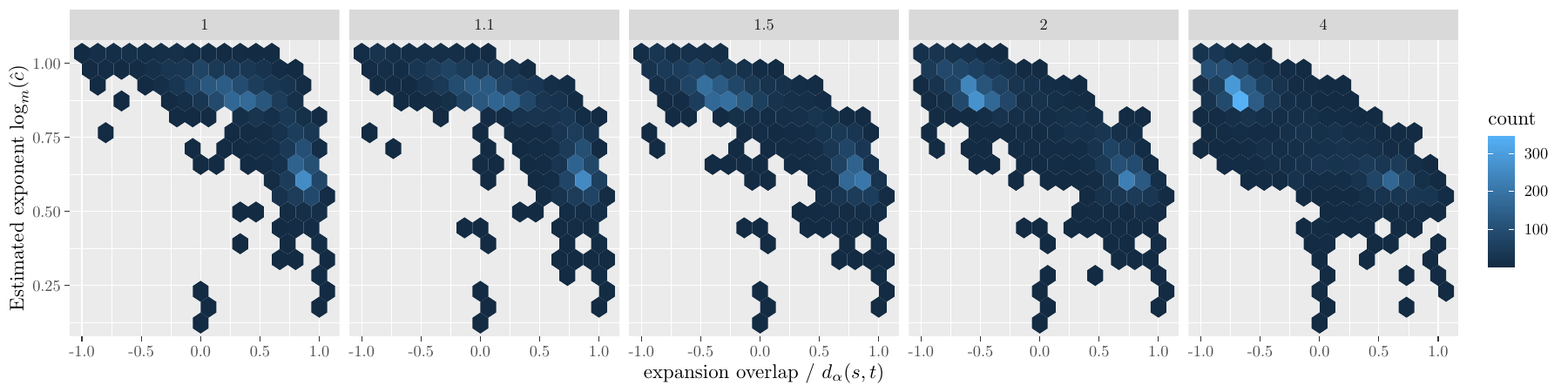}
    \caption{Estimated exponent for parameter $c$ of \cref{lem_sublinear_with_cheap_start} for $α = 0.1$ and different values of $b$.}
    \label{fig_different_b}
\end{figure}

Similarly, \cref{fig_different_b} shows how different values for the minimum base of the expansion $b$ influence the results only slightly as long as $b$ is sufficiently far from $1$.
We again plot the same distribution as above, this time for different values of $b$ between $1$ and $4$.
If $b$ is set to $1$, then it suffices for two consecutive exploration steps to have non-decreasing cost in order for them to be considered as $b$-expanding.
This increases the lengths of the expansion overlaps slightly, and shifts the plotted distribution to the right, leading to more vertices that fulfil the condition of \cref{lem_sublinear_with_cheap_start}.
However it is arguably not very reasonable to expect to observe asymptotic behavior on instances of rather small constant size, if exponential growth is allowed a base too close to $1$.
For only slightly larger bases such as $1.1$ and even more so $2$ and $4$, the correspondence of theoretical predictions and empirical observations becomes much better.
This justifies our choice of $b=2$ and $α=0.1$ in the main part of this paper not just for the sake of simplicity and clarity.

\end{document}